\documentclass[11pt]{article}
\usepackage{amsmath,amsfonts}
\usepackage{hyperref}
\usepackage{graphicx}
\usepackage{amsthm}

\usepackage[margin=1in]{geometry}
\usepackage[font=small,labelfont=bf]{caption}
\usepackage[labelformat=simple]{subcaption}

\usepackage{libertine}\usepackage[libertine]{newtxmath}
\usepackage[scaled=0.96]{zi4}

\usepackage{xcolor}

\newcommand\NAE{\operatorname{NAE}}

\newtheorem{theorem}{Theorem}

\newtheorem{lemma}[theorem]{Lemma}

\theoremstyle{remark}

\theoremstyle{definition}

\let\epsilon=\varepsilon
\def\defn#1{\textbf{\textit{\boldmath #1}}}

\author{Erik D. Demaine\thanks{Computer Science and Artificial Intelligence Laboratory, Massachusetts Institute of Technology, Cambridge, MA 02139, USA. \protect\url{edemaine@mit.edu} }
\and Kritkorn Karntikoon\thanks{Department of Computer Science, Princeton University, Princeton, NJ 08544, USA. \protect\url{kritkorn@princeton.edu}} \and Nipun Pitimanaaree\thanks{Alpha Finance Lab, Bangkok, Thailand. \protect\url{nipun@alphafinance.io}}}
\title{2-Colorable Perfect Matching is NP-complete in \\ 2-Connected 3-Regular Planar Graphs}
\date{}

\begin{document}

\maketitle
\begin{abstract}
    The 2-colorable perfect matching problem asks whether a graph can be
colored with two colors so that each node has exactly one neighbor
with the same color as itself.
We prove that this problem is NP-complete,
even when restricted to 2-connected 3-regular planar graphs.
In 1978, Schaefer proved that this problem is NP-complete in
general graphs, and claimed without proof that the same result holds when
restricted to 3-regular planar graphs.
Thus we fill in the missing proof of this claim,
while simultaneously strengthening to 2-connected graphs
(which implies existence of a perfect matching).
We also prove NP-completeness of $k$-colorable perfect matching,
for any fixed $k \geq 2$.

\end{abstract}

\section{Introduction}\label{section-intro}
Schaefer's \textit{The Complexity of Satisfiability Problems}
\cite{schaefer1978complexity} is a seminal paper in NP-completeness,
presenting his famous dichotomy theorem characterizing the complexity of
satisfiability problems.  It introduced and proved NP-complete two new problems
--- One-in-Three Satisfiability (a.k.a.\ Positive 1-in-3SAT) and
Not-All-Equal Satisfiability (a.k.a.\ Positive NAE 3SAT) ---
that have served as the basis for countless reductions since.

Schaefer's paper also introduced and proved NP-complete a third lesser-known
problem, called \defn{2-colorable perfect matching}:
given a graph $G$, can the nodes of $G$ be colored with two colors
so that each node has exactly one neighbor the same color as itself?
This problem served as an example application of NAE 3SAT: Schaefer showed
a fairly simple reduction from NAE 3SAT to 2-colorable perfect matching
in general graphs.  Then he made a tantalizing claim:

\begin{center}
  \includegraphics[width=0.75\textwidth]{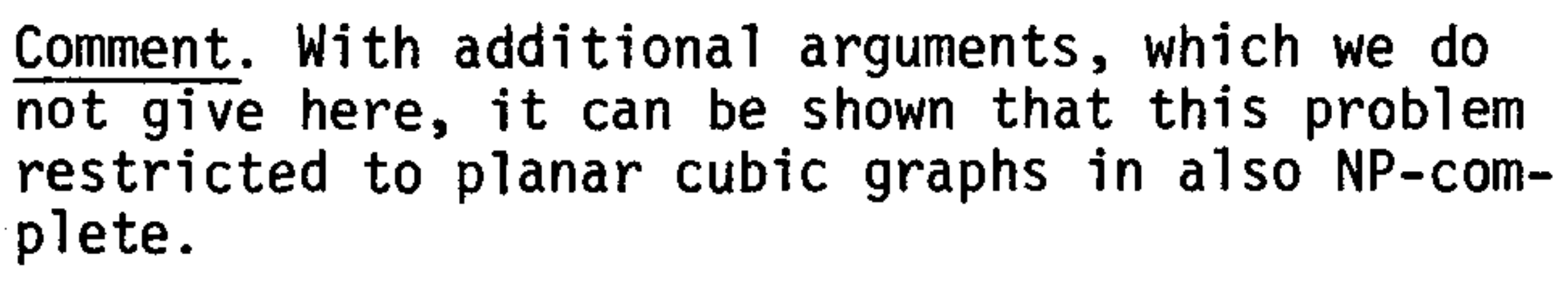}
\end{center}


\noindent
As far as we know, this claim has never been proved in the literature.

In this paper, we fill in this gap left by Schaefer's paper,
by proving that 2-colorable perfect matching in 3-regular planar graphs
is NP-complete.
We also extend the result in two ways:
by restricting the graph to be 2-connected,
and by generalizing to $k$-colorings for $k \geq 2$.
Adding 2-connectivity is particularly interesting for this problem
because 2-connected 3-regular graphs always have a perfect matching
\cite{petersen1891}, which can be found in near-linear time and, when
further restricted to planar graphs, in linear time \cite{MatchingJAlgorithms}.
Thus the graphs we consider always have perfect matchings, but finding a
2-colorable perfect matching remains NP-complete.

The rest of this paper is organized as follows.
In Section~\ref{section-NAE}, we present the NAE 3SAT problem and prove
NP-completeness of a stronger form called Clause-Connected Positive NAE 3SAT,
Section~\ref{section-planar} then shows our main reduction from this problem
to 2-colorable perfect matching in 2-connected 3-regular planar graphs.
Then Section~\ref{section-k_colors} proves NP-hardness of
$k$-colorable perfect matching
(without the 2-connectivity, 3-regularity, or planarity restrictions).

\section{Clause-Connected Positive NAE 3SAT}\label{section-NAE}
In this section, we define and prove NP-hard a problem called
``Clause-Connected Positive NAE 3SAT'', which is what we reduce from
to prove NP-hardness of 2-colorable perfect matching in
2-connected 3-regular planar graphs.

Recall that \defn{Positive%
\footnote{We add the ``Positive'' adjective to make clear that the problem
  does not allow negated literals of the form $\neg x_i$ to appear in a clause.
  Schaefer's original definition of Not-All-Equal Satisfiability
  also did not allow negations.}
NAE 3SAT} \cite{schaefer1978complexity} asks whether variables
$x_1, x_2, \dots, x_n$ can be assigned binary values
(say, among \textsc{true} or \textsc{false}) in order to satisfy
given clauses $c_1, c_2, \dots, c_m$,
where each clause is of the form $\NAE(x_i, x_j, x_k)$,
meaning that $x_i$, $x_j$, and $x_k$ must not be assigned all the same value.



\defn{Clause-Connected Positive NAE 3SAT} further requires that the variables from each clause be connected via shared variables. More formally, if we construct a graph $G$ where each node represents each clause, and two nodes are connected if the corresponding clauses share a variable, then the graph $G$ has to be connected.

\begin{theorem} Clause-Connected Positive NAE 3SAT is NP-complete.
\end{theorem}

\begin{proof}
We prove NP-hardness by reducing from Positive NAE 3SAT.
Suppose we are given a Positive NAE 3SAT instance $(V, C)$, where $V$ is a set of variables and $C$ is a set of clauses appearing in a formula of the instance.
Our reduction produces an instance of Clause-Connected Positive NAE 3SAT as follows:
\begin{enumerate}
    \item Start with the same formula $(V, C)$.
    \item If the clauses are not connected, then there are at least two distinct clauses, $c_1$ and $c_2$, with no variable in common. Then we add a new NAE clause consisting of one variable from $c_1$, one variable from $c_2$, and one new variable~$v'$. For example, if there are two clauses, $\NAE(v_1, v_2, v_3)$ and $\NAE(v_4, v_5, v_6)$ where all $v_i$s are different, then we add a new clause $\NAE(v_1, v_4, v')$ to the formula.
    \item Repeat the previous step until the clauses are connected. Because we reduce the number of connected components of clauses in each step (and the newly introduced variable is always connected), the algorithm will terminate in $O(|C|)$ steps.
\end{enumerate}
The result of this algorithm is clearly an instance of Clause-Connected Positive NAE 3SAT.
To show that this reduction is correct, we need to show that a solution to either of the instances implies the existence of a solution to the other instance.

Suppose first that there is a solution to the Positive NAE 3SAT instance. Then the same assignment to the original variables will satisfy all original clauses. For each new NAE clause created in Step~2, we can assign the truth value to the new variable $v'$ to be the opposite of one of the other variables, thereby satisfying that clause no matter how the other two variables were already assigned. Because this variable does not appear in any other clause, we obtain a satisfying assignment.

On the other hand, if there is a solution to the Clause-Connected Positive NAE 3SAT instance, then this assignment satisfies the original clauses from the Positive NAE 3SAT instance as well. This is because the set of clauses in the Positive NAE 3SAT instance is a subset of the clauses in the Clause-Connected Positive NAE 3SAT instance.
\end{proof}

Note that this reduction is not parsimonious, i.e.,
it does not preserve the number of solutions.


\section{Planar 2-Connected 3-Regular 2-Colorable Perfect Matching}\label{section-planar}
In this section, we prove NP-hardness of 2-colorable perfect matching when restricted to 2-connected 3-regular planar graphs, by a reduction from Clause-Connected Positive NAE 3SAT (from Section~\ref{section-NAE}).
Overall, we follow the structure of the reduction from NAE 3SAT to 2-colorable perfect matching (in general graphs) sketched in \cite{schaefer1978complexity},
which builds a clause gadget to enforce NAE constraints and an equal-colors gadget to make copies of variables.
Our reduction uses the same clause gadget
(described in Section~\ref{sec:clause}),
but differs everywhere else:
we construct a \emph{different}-colors (negation) gadget
in Section~\ref{sec:different} to achieve 3-regularity;
a degree-reduction gadget in Section~\ref{sec:degree}
to achieve 3-regularity,
and a crossover gadget in Section~\ref{sec:crossover} to achieve planarity.
Reducing from Clause-Connected Positive NAE 3SAT lets us achieve 2-connectivity.


\subsection{Clause Gadget}\label{sec:clause}
The \defn{clause gadget} is a complete graph $K_4$ on four nodes,
as shown in Figure~\ref{figure: clause gadget},
For a clause $\NAE(x,y,z)$, three of the nodes are \defn{variable} nodes,
labeled $x$, $y$, and $z$; while the fourth (middle) node is a \defn{clause}
node unique to the clause, so we label it $(x,y,z)$.
By coloring this graph with the colors ``$0$'' and ``$1$'',
each variable is assigned truth values according to the color.
The requirement of 2-colorable perfect matching forces node $(x,y,z)$
to match with exactly one of the three variable nodes,
forcing the other two variable nodes to have the opposite color
from this matched pair (and thus match with each other).
Therefore, 2-colorable perfect matchings correspond exactly to
not-all-equal colorings of the variable nodes $x,y,z$,
where the clause node always gets the minority color among $x,y,z$.

\begin{figure}[h]
    \centering
    \includegraphics[width=0.25\textwidth]{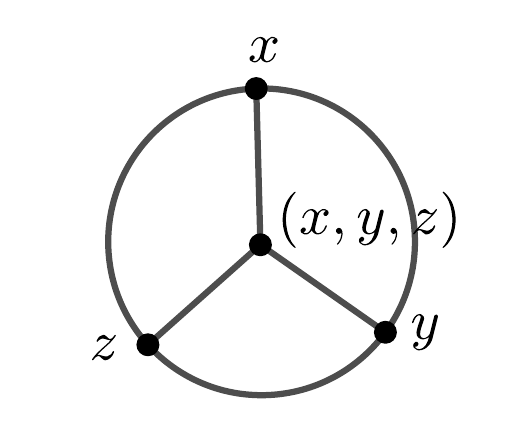}
    \caption{Clause gadget for $\NAE(x,y,z)$, consisting of three variable nodes $x,y,z$ and one clause node $(x,y,z)$ connected in a clique.}
    \label{figure: clause gadget}
\end{figure}

\subsection{Different-Colors Gadget}\label{sec:different}
The \defn{different-colors gadget} ensures that two nodes $x,y$ have different colors.
The two nodes $x,y$ are connected together by a 3-regular planar graph
as shown in Figure~\ref{figure: different color},
which also shows what we claim is the unique 2-coloring up to reflection.
First we show that the coloring of the central square is forced,
up to reflection:
\begin{itemize}
\item it cannot repeat a color more than twice in cyclic order,
  or else a middle node would have too many same-colored neighbors;
\item it cannot be alternating,
  because then some color would lack a same-color neighbor; and
\item it cannot have monochromatic top and bottom edges,
  because then some color would have an extra same-color neighbor;
\item so it must have monochromatic left and right edges.
\end{itemize}
Then the remaining vertices' colors are forced from the inside out.

This gadget has three additional properties that we will need later:
\begin{enumerate}
\item All eight internal nodes in the gadget are matched with each other, while the nodes $x$ and $y$ do not match with them. This will let us use this gadget in order to increase the degree of nodes $x$ and~$y$.
\item As a replacement for an edge $(x,y)$, the gadget preserves the properties of being 3-regular, planar, and 2-connected.
\item The gadget contains a cut (drawn brown dotted in Figure~\ref{figure: different color}) such that every edge across the cut has differently colored endpoints. We call this cut the \defn{different-colors cut}.
\end{enumerate}

\begin{figure}
    \centering
    \includegraphics[width=0.5\textwidth]{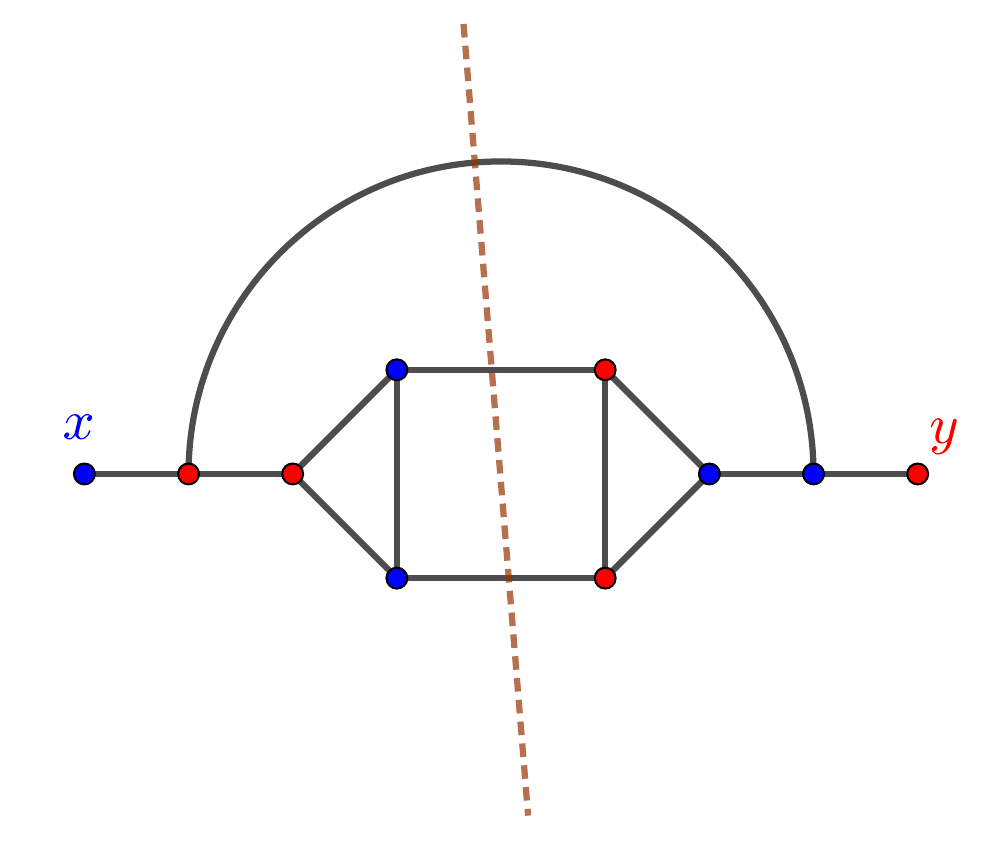}
    \caption{Different-colors gadget, which connects nodes $x$ and $y$
      and ensures that $x$ and $y$ have different colors.
      The brown dotted line is a different-colors cut.}
    \label{figure: different color}
\end{figure}

\subsection{Degree-Reduction Gadget}\label{sec:degree}
The \defn{degree-reduction gadget} reduces the degree of a
high-degree node~$x$, as shown in Figure~\ref{figure: degree reduction}:
if $x$ has $k > 3$ neighbors $y_1, y_2, \dots, y_k$,
then we split $x$ into $x_1, x_2, \dots, x_k$,
where each $x_i$ has an edge to $y_i$,
and connect consecutive $x_i$s by a path of length~$4$,
resulting in an overall $x$ path of length $4(k-1)$.
Notably, the intermediate nodes in the path have degree~$2$.

\begin{figure}[h]
    \centering
    \includegraphics[width=0.75\textwidth]{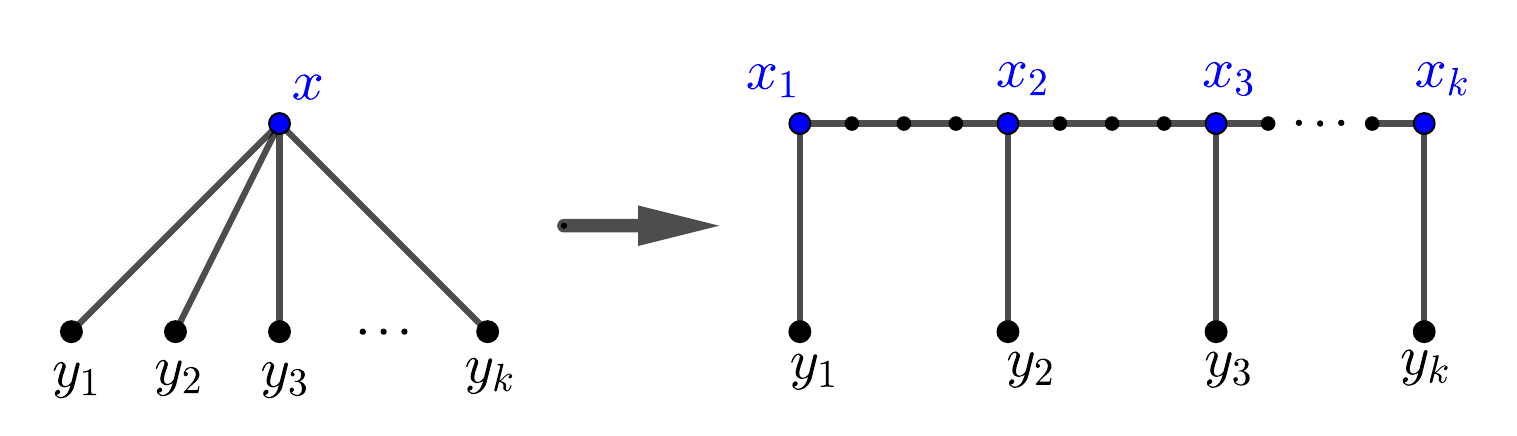}
    \caption{Degree-reduction gadget. For a node $x$ of degree $k > 3$, we create $k$ copies $x_1, x_2, \ldots, x_k$ in a length-$4(k-1)$ path. This gadget ensures that $x_i$s have the same color and exactly one is matched with a neighboring $y_i$.}
    \label{figure: degree reduction}
\end{figure}

We claim that this gadget must match exactly one edge $(x_i, y_i)$.
First, it must match at least one such edge,
because the $x$ path has an odd number of nodes.
Second, if it matched two such edges, consider the nearest pair
$(x_i, y_i)$ and $(x_j, y_j)$; the subpath strictly between $x_i$ and $x_j$
has an odd number of nodes, so it cannot have a perfect matching.
Therefore there is exactly one matched edge $(x_i, y_i)$.

Indeed, any choice of matched edge $(x_i, y_i)$ fixes the parity problem,
splitting the $x$ path into two subpaths with an even number of vertices,
enabling a unique perfect matching.
Consecutive $x_j$s have exactly one matched edge along the $x$ path in between
them, so all $x_j$s have the same color; the matched $y_i$ must also have the
same color; and the remaining $y_j$s must have the opposite color.
Therefore the gadget behaves just like the original $x$ vertex.

Similar to the different-colors gadget, this gadget preserves the properties
of being planar and 2-connected.
But it introduces vertices of degree $2$, so it does not preserve 3-regularity.
We will fix these degrees later,
in a degree-expansion phase in Section~\ref{subsection: reduction}.

\subsection{Crossover Gadget}\label{sec:crossover}
The \defn{crossover gadget} lets us remove a crossing between two edges
$x y$ and $z w$, provided we want both of those edges to be bichromatic
(i.e., $x$ and $y$ are colored differently, as are $z$ and $w$).
As shown in Figure~\ref{figure: crossover gadget},
we construct four new nodes $x', y', z', w'$ connected in a cycle;
and we connect $x$ to $y'$, $y$ to $x'$, $z$ to $w'$, and $w$ to $z'$
using different-colors gadgets (drawn as $\neq$ circles in the figure). 

\begin{figure}[h]
    \centering
    \includegraphics[width=0.75\textwidth]{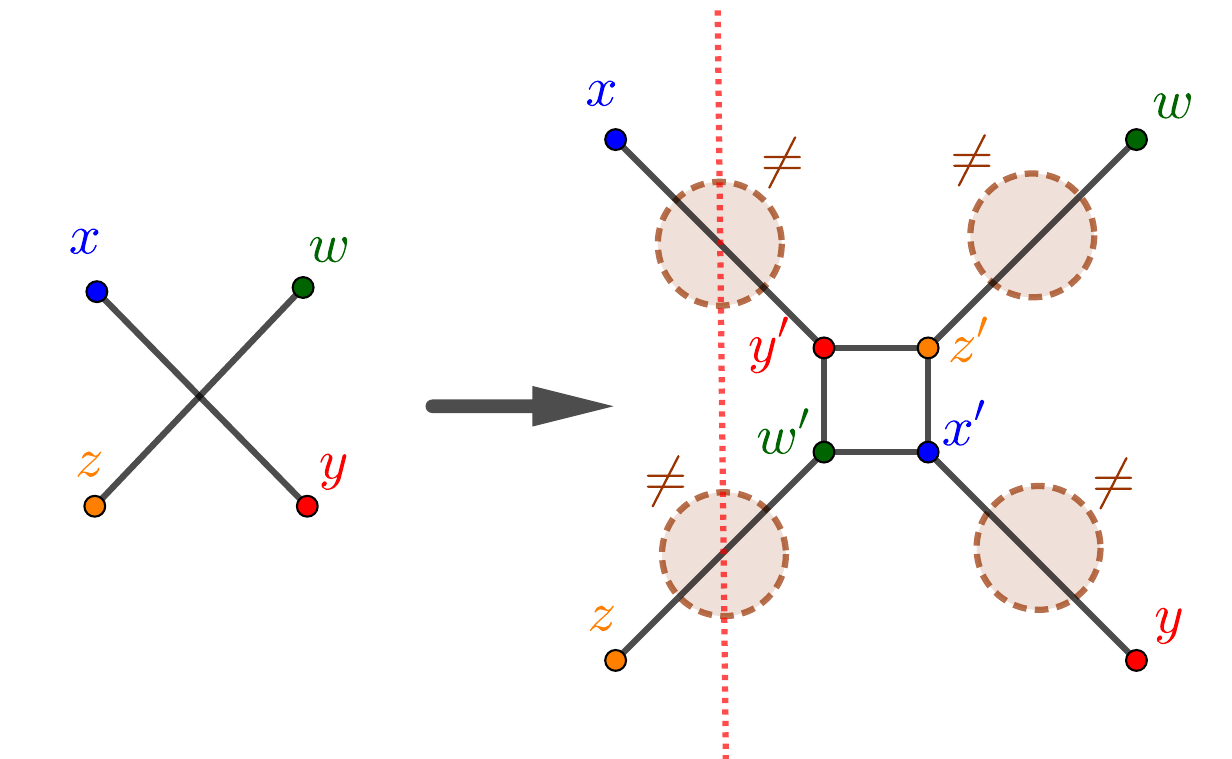}
    \caption{Crossover gadget,
      which uncrosses two bichromatic edges $x y$ and $z w$.
      This gadget uses four different-colors gadgets
      from Figure~\ref{figure: different color},
      represented by $\neq$ circles,
      to connect original nodes $x,y,z,w$
      with four new nodes $x', y', z', w'$.
      The colors in the figure satisfy that red and blue are opposite,
      and green and yellow are opposite, with an unspecified
      correspondence between \{red, blue\} and \{green, yellow\}.
      The red dotted line is a different-colors cut.}
    \label{figure: crossover gadget}
\end{figure}

Because the different-colors gadget does not match its endpoints,
the central square $x', y', z', w'$ must be matched internally.
Thus $x'$ is the same color as either $z'$ or~$w'$,
and opposite both $y'$ and the other of $w'$ or~$z'$.
By the different-colors gadgets, the colors satisfy $x \neq y' \neq x'$,
so $x$ and $x'$ must have the same color;
and symmetrically so do $y$ and $y'$, $z$ and $z'$, and $w$ and~$w'$.
Again by the different-colors gadgets, $x$ and $y$ must have opposite colors,
as must $z$ and~$w$.
By choosing the matching in the central square
(monochromatic top and bottom edges, or monochromatic left and right edges),
we can freely choose the coloring of $x \neq y$ independent of $z \neq w$.

Similar to the different-colors gadget,
this crossover gadget has the following properties:
\begin{enumerate}
\item All four internal nodes $x', y', z', w'$ are matched with each other,
  so there will be no effect on the original graph in terms of matching.
\item The gadget preserves the properties of being 3-regular and 2-connected,
  and introduces no new crossings.
\item The gadget contains a different-colors cut
  (as defined in Section \ref{sec:different},
  and drawn red dotted in Figure~\ref{figure: crossover gadget}):
  we can combine the two cuts from the different-colors gadgets
  between $x$ and $y'$ and between $z$ and~$w'$.
\end{enumerate}

\subsection{Reduction from Clause-Connected Positive NAE 3SAT}\label{subsection: reduction}
Now that we have the necessary gadgets, we describe our overall reduction
from Connected Positive NAE 3SAT to
2-colorable perfect matching on 2-connected 3-regular planar graphs.
We construct the desired graph in the following steps:
\begin{enumerate}
\item Create clause gadgets representing the NAE 3SAT instance.
\item Connect shared variables via a chain of two different-colors gadgets.
\item Duplicate the entire graph to make the graph 2-connected.
\item Decrease degrees $>3$ via the degree-reduction gadget.
\item Increase degrees from $2$ to $3$.
\item Add crossovers to make the graph planar.
\end{enumerate}

First, for each clause in the NAE 3SAT instance,
we create a clause gadget corresponding to it,
as described in Section~\ref{sec:clause} and Figure~\ref{figure: clause gadget}.
Second, when two nodes $v_1$ and $v_2$ of two clause gadgets
are labeled with the same variable,
we connect them with a chain of two different-colors gadgets,
as shown in Figure~\ref{figure: step2 connect variables}
(where a $\neq$ circle represents a different-colors gadget). 

\begin{figure}[h]
    \centering
    \includegraphics[width=0.5\textwidth]{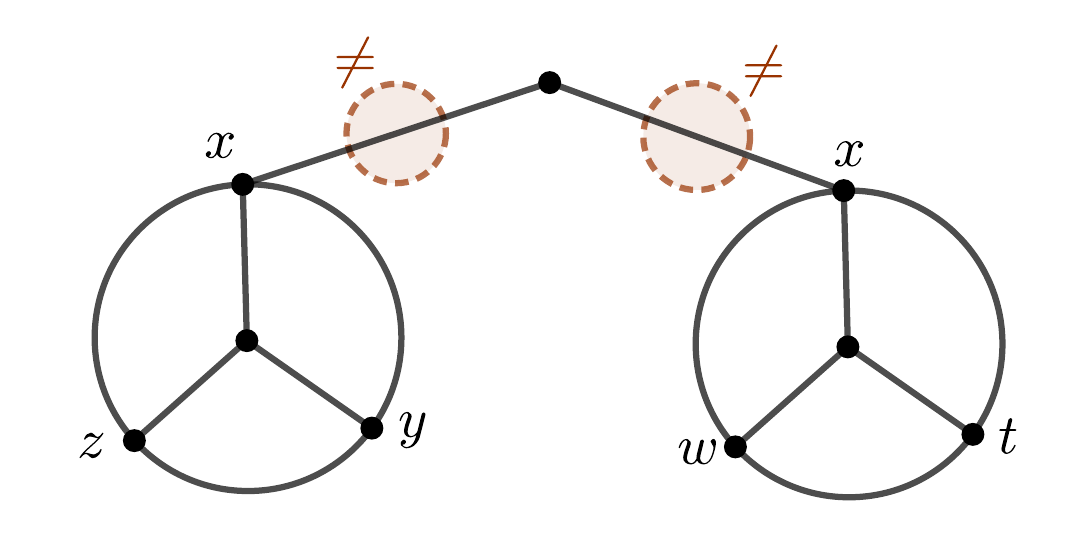}
    \caption{How we connect two nodes from different clause gadgets that correspond to the same variable~$x$: using a chain of two different-colors gadgets, represented by $\neq$ circles.}
    \label{figure: step2 connect variables}
\end{figure}

Third, we duplicate the entire graph,
and connect each variable node with its clone by a different-colors gadget,
as shown in Figure~\ref{figure: step3 copy graph}
(where the different-colors gadgets are represented by orange dotted lines).
In this step, we do not connect other (nonvariable) nodes with their clones,
in order to preserve their degrees.
This step achieves 2-connectivity:

\begin{figure}[h]
    \centering
    \includegraphics[scale=0.4]{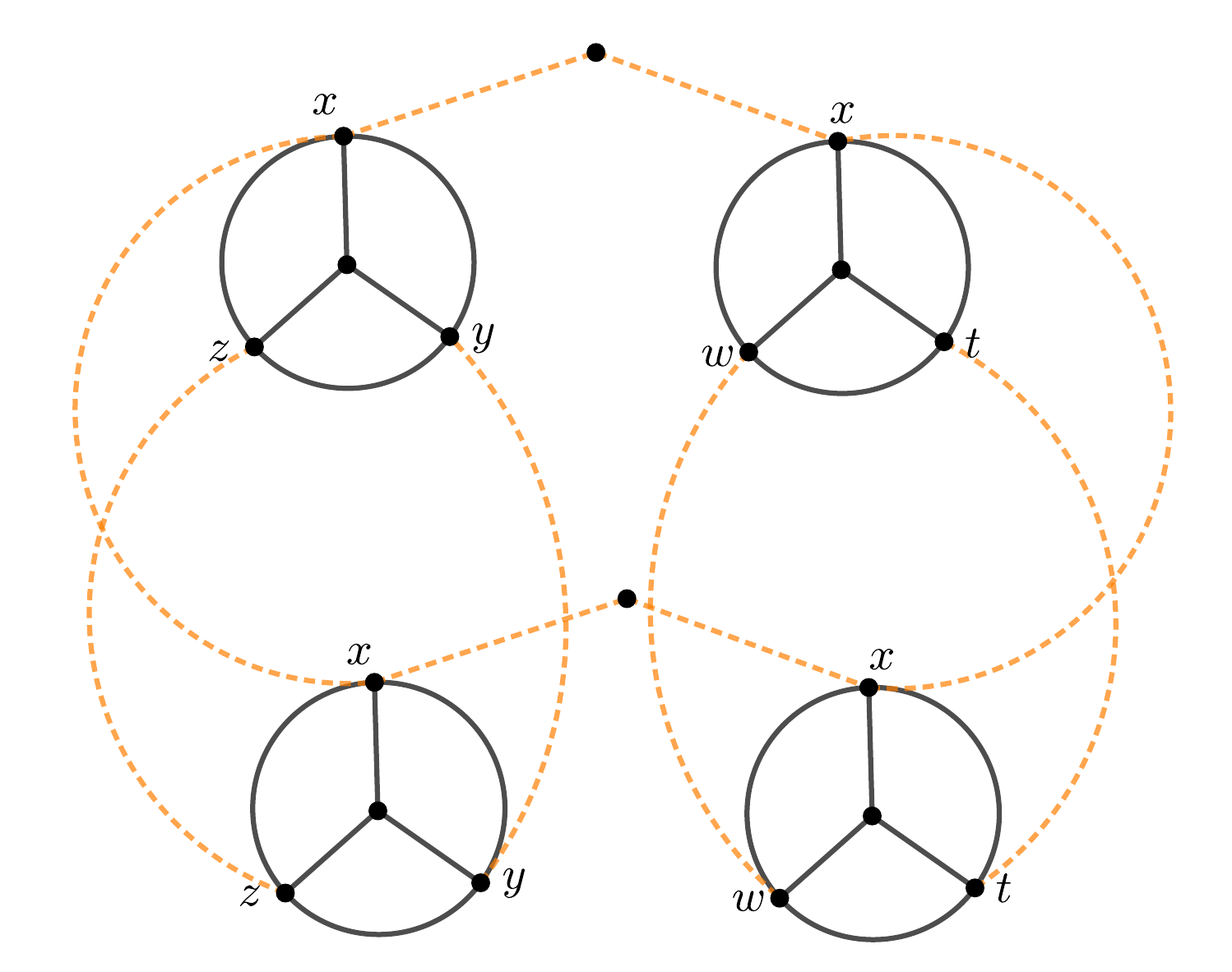}
    \caption{Duplicating the entire graph, and connecting each variable node with its clone by a different-colors gadget (represented by an orange dotted line).}
    \label{figure: step3 copy graph}
\end{figure}

\begin{lemma}
  This duplicated graph is 2-connected.
\end{lemma}
\begin{proof}
It suffices to show that, for any pair of nodes $(s,t)$ in the graph,
there are two interior-vertex-disjoint paths that connect them. 

First, consider the case where both $s$ and $t$ are original nodes
(or symmetrically, both clone nodes);
refer to Figure~\ref{figure: 2-connected proof a}.
By clause connectivity and the construction,
there is a path between $s$ and $t$ within the
original graph; take the shortest such path as the first path
(drawn red in the figure).

\begin{figure}[h]
    \centering
    \subcaptionbox
      {\label{figure: 2-connected proof a}
        Between two nodes in the same copy of the graph.}
      {\includegraphics[width=0.5\textwidth]{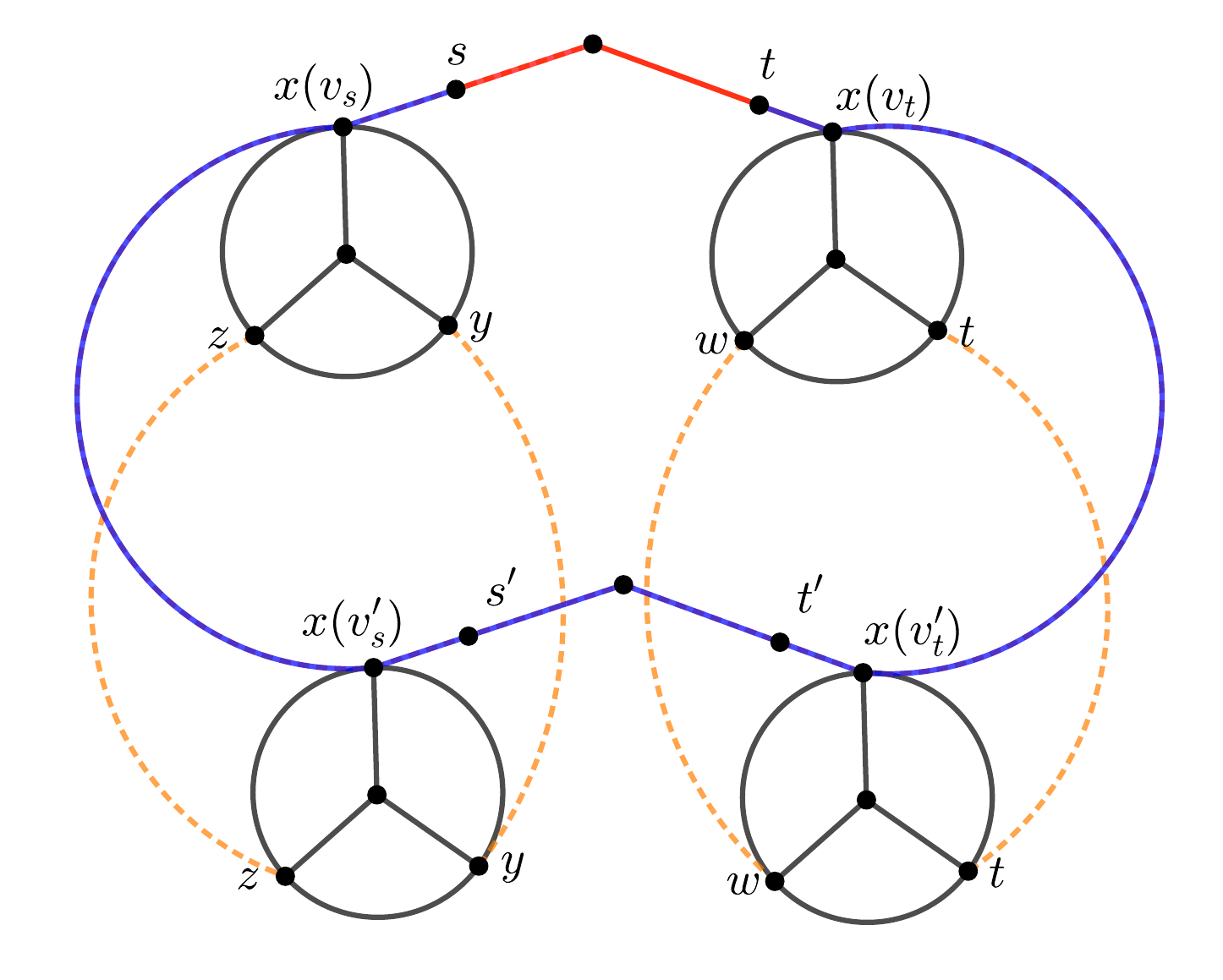}}%
    \hfill
    \subcaptionbox
      {\label{figure: 2-connected proof b}
        Between two nodes in opposite copies of the graph.}
      {\includegraphics[width=0.5\textwidth]{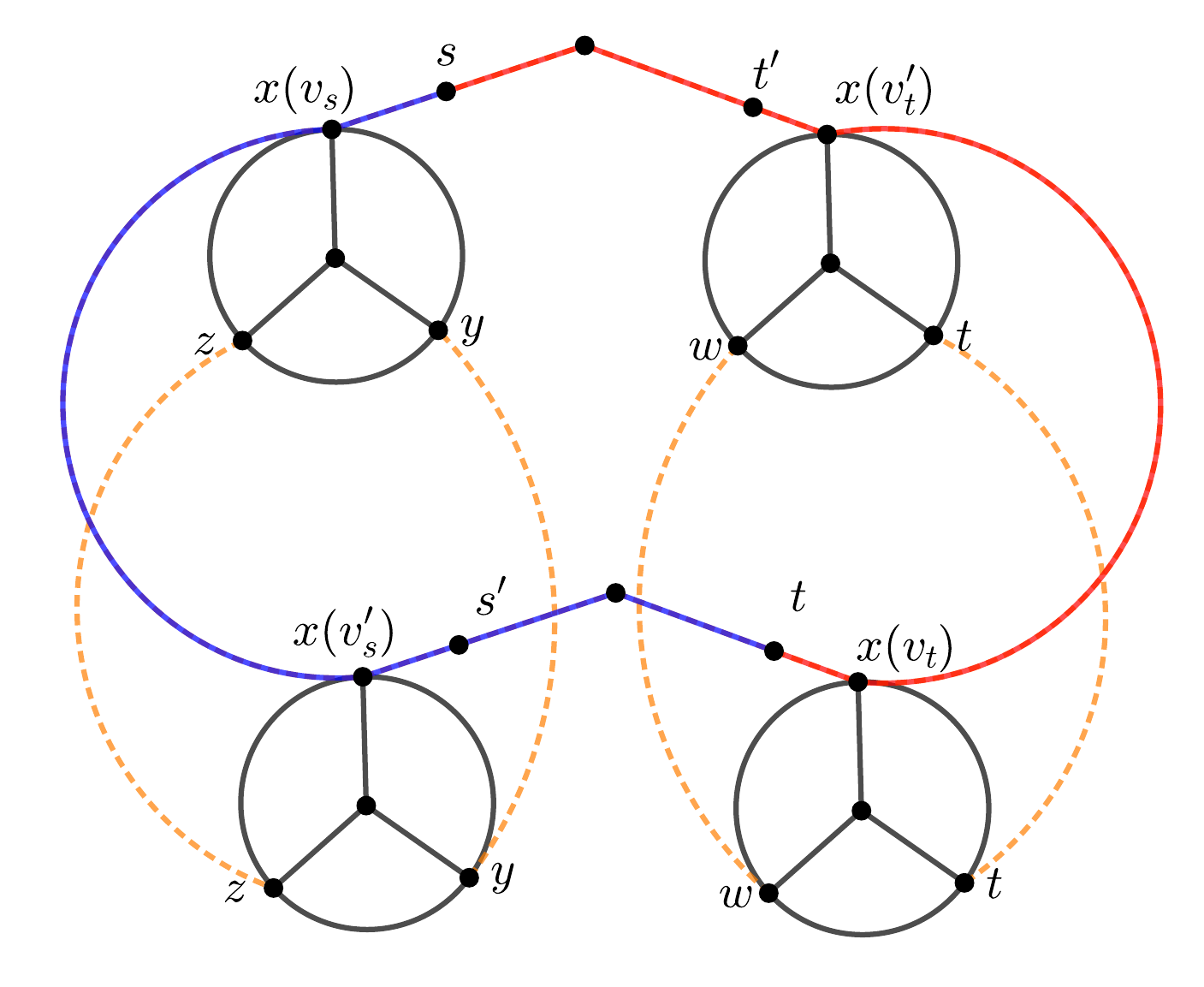}}%
    \caption{Constructing two interior-vertex-disjoint paths (red and blue) between two nodes $s$ and~$t$.}
    \label{figure: 2-connected proof} 
\end{figure}

Now we construct a variable node $v_s$ such that $s$ can travel to $v_s$
without intersecting this shortest path.
If $s$ is a variable node, then take $v_s$ to be $s$ itself.
If $s$ is a clause node, then one of its three neighbors
will not be used in the shortest path
(otherwise, the shortest path could be further shortened).
If $s$ is a part of a different-colors gadget,
then one of the variable endpoints will not be used
(otherwise, the shortest path could again be shortened).

Similarly, we can construct $v_t$ for $t$.
If $v_s = v_t$, then we found another path between $s$ and $t$, as desired.
Otherwise, we pick the following path (drawn blue in the figure):
the path from $s$ to $v_s$,
the edge from $v_s$ to $v_s$'s clone,
the shortest path from $v_s$'s clone to $v_t$'s clone within the clone graph,
the edge from $v_t$'s clone to $v_t$, and
the path from $v_t$ to $t$.
This path will not intersect the shortest path in the original graph,
because the remaining part is in the clone graph.

Second, consider the case where $s$ is an original node and
$t$ is a clone node (or symmetrically, the reverse);
refer to Figure~\ref{figure: 2-connected proof b}.
As before, we construct $v_s$ and $v_t$, but now relative to a shortest path
from $s$ to $t$'s clone (which is in the original graph).
Then we construct the first path (drawn blue in the figure)
by combining the shortest path from $s$ to $t$'s clone,
the path from $t$'s clone to $v_t$'s clone,
the edge from $v_t$'s clone to $v_t$, and
the path from $v_t$ to $t$.
We construct the second path path (drawn red in the figure) by combining
the shortest path from $t$ to $s$'s clone within the clone graph,
the path from $s$'s clone to $v_s$'s clone,
the edge from $v_s$'s clone to $v_s$, and
the path from $v_s$ to~$s$.
\end{proof}

Fourth, we use the degree-reduction gadget to decrease the degrees
of all nodes that have degree more than~$3$.
By our construction (in particular the duplication in Step~3),
these nodes consist of only variable nodes (within clause gadgets)
and their clones.
Thus we can embed the degree-reduction gadgets within the clause gadgets.
Figure~\ref{figure: step5 adjust degree} sketches the result of this step,
but only draws a few of the degree-2 vertices in the degree-reduction gadgets.

\begin{figure}[h]
    \centering
    \includegraphics[scale=0.4]{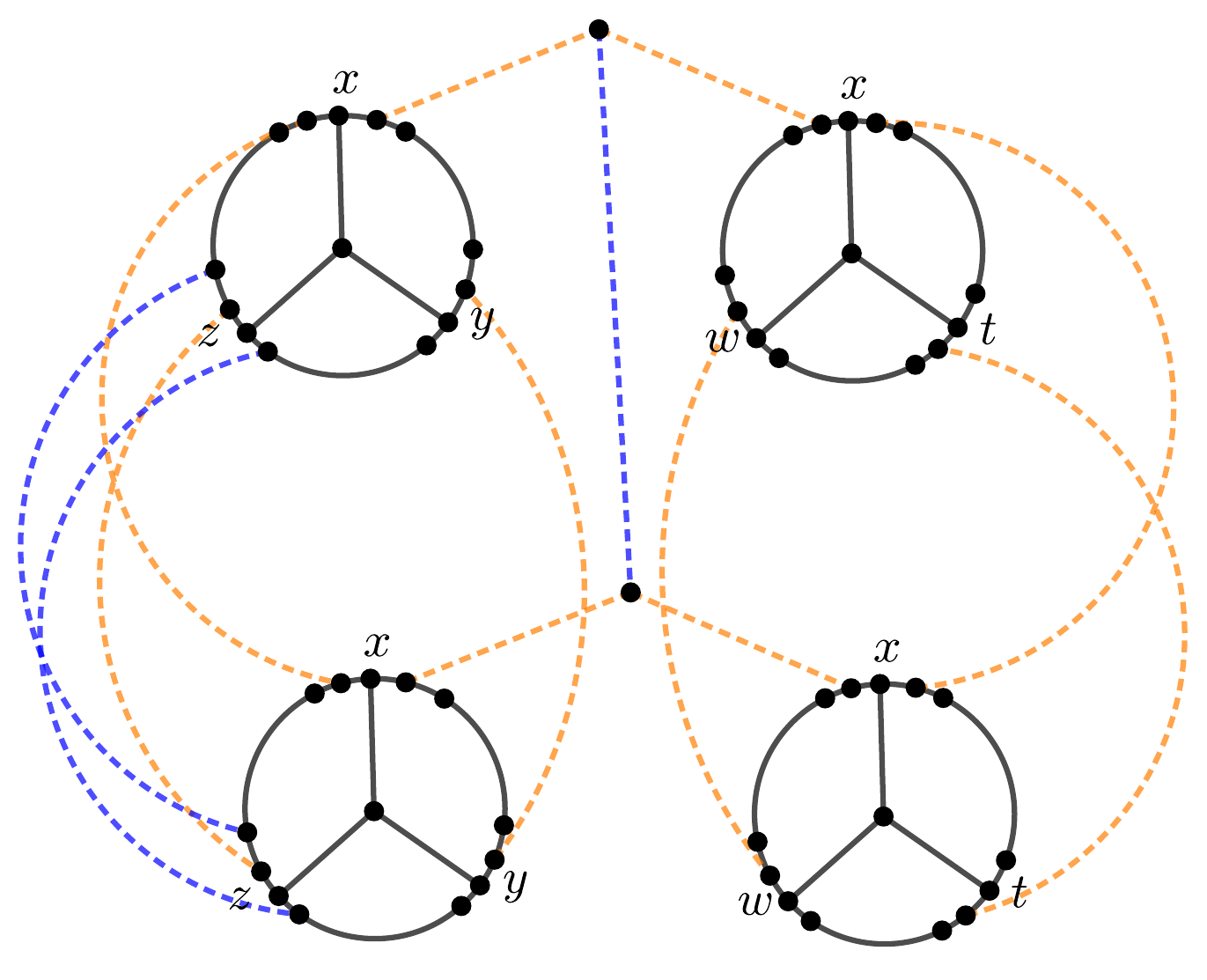}
    \caption{Making the graph 3-regular, using degree-reduction gadgets to split high-degree variable nodes, and different-colors gadgets (drawn as blue dotted lines) between the two copies to raise degrees from 2 to 3, in particular for nodes in degree-reduction gadgets (not all of which are drawn).}
    \label{figure: step5 adjust degree}
\end{figure}

Fifth, we increase the degree of each node that has degree less than~$3$,
which by 2-connectivity must have degree~$2$,
by connecting the node to its clone with a different-colors gadget.
Figure~\ref{figure: step5 adjust degree} draws the different-colors gadgets
added in this step as blue dotted lines
(but there are more than the three drawn connecting lines,
because there are more degree-2 nodes than drawn).
Each node and its clone are processed similarly, so we maintain that they
have the same degree in every step.
Thus, connecting a node of degree less than $3$ with its clone
will increase the degrees for both of them together,
eventually resulting in a 2-connected 3-regular graph.

Finally, in Step~6, we make the graph planar using crossover gadgets.
This requires some care, because it can only cross pairs of bichromatic edges.
Because all ``long distance'' connections between clause gadgets
(including embedded degree-reduction gadgets)
are made via different-colors gadgets,
the only possible intersections are between pairs of different-colors gadgets.
These gadgets have different-colors cuts
(as defined in \ref{sec:different}).
We can arrange for crossings to occur only at the meeting of these cuts
for the two gadgets, so that the crossing edges are bichromatic,
enabling us to use a crossover gadget to remove the crossing.
After we place some crossover gadgets, some crossings between
different-colors gadgets may now instead involve crossover gadgets,
but because crossover gadgets also have different-colors cuts,
we can continue the same construction.
By repeatedly applying a crossover gadget in this way, we remove all crossings.
Because the crossover gadget preserves the properties of being 2-connected
and 3-regular, we obtain a 2-connected 3-regular planar graph, as desired.

\begin{theorem}
  2-colorable perfect matching in 2-connected 3-regular planar graphs is NP-complete.
\end{theorem}
\begin{proof}
It is easy to see that this problem is in NP because we can verify whether a given coloring satisfies our requirements. For NP-hardness, it remains to prove correctness of the reduction above from Clause-Connected Positive NAE 3SAT.
We have already argued that, by following the steps above, we create a 2-connected 3-regular planar graph, so the instance is valid.
It remains to show that a true instance of Clause-Connected Positive NAE 3SAT
is mapped to a true instance of 2-colorable perfect matching, and vice versa.

Given a true instance of Clause-Connected Positive NAE 3SAT, we color each variable node based on the truth value of the corresponding variable. For each clause gadget, there are two nodes that have the same color, and so we pair another node with the center of the clause gadget. Thus, we have already paired all the variable nodes. Then we can color the other nodes following the gadget constructions. After that, we color each node's clone with its negated color (because they are connected by a different-colors gadget). Because of the properties of each gadget, the nonvariable nodes will be paired to each other without violating any requirement. Thus, we have a valid coloring and so a true instance of 2-colorable perfect matching.

Conversely, given a true instance of 2-colorable perfect matching, we assign the truth value of each variable based on the color of the corresponding variable node in the primary copy of the graph (ignoring the clone copy). There are many variable nodes that correspond to the same variable, but these nodes will be colored with the same color because we connected them with a chain of two opposite-colors gadgets. Because all nonvariable nodes within each gadget are forced to pair with one another, variable nodes and clause nodes are forced to pair with one another as well. Considering a clause gadget, two of the three variable nodes then have to pair to each other, and the remaining node will pair with the clause node, using the opposite color. Thus, the three values are not equal to each other for every clause, resulting in a valid assignment for the NAE 3SAT problem. 

Therefore, our reduction is correct, and so 2-colorable perfect matching on 2-connected 3-regular planar graph is NP-complete as desired.
\end{proof}

\section{\texorpdfstring{\boldmath $k$}{k}-Colorable Perfect Matching}\label{section-k_colors}
In this section, we prove that $k$-colorable perfect matching is NP-complete
(without the additional properties of 2-connected, 3-regular, or planar).
We follow the same proof structue as Section~\ref{section-planar},
but use a new gadget to constrain the coloring.

\subsection{\texorpdfstring{\boldmath $(k-2)$}{(k-2)}-Color Gadget}

A \defn{$(k-2)$-color gadget} contains exactly $k-2$ colors, and
so that we can locate $k-2$ nodes that all have different colors.
As shown in Figure~\ref{figure: k-color gadget}, the gadget consists of
a path of $2k-4$ nodes, with every pair of even nodes connected together,
resulting in a path through a clique of size $k-2$.
Assume that the rest of the graph is attached to this gadget only at
even nodes, so the odd nodes are exactly as in the figure.

To show that the gadget consists of exactly $k-2$ colors,
first consider the leftmost node.
It has degree $1$, so it must match with its neighbor, with some color~$x_1$.
Then, at step $i$, the $(2i-1)$st node --- which is odd so has degree $2$ ---
must match with the $(2i)$th node (as the $(2i-2)$nd node is already matched),
with some color $x_i$.
Because each node must have exactly one neighbor with the same color,
and all $x_i$s are used in the clique of even nodes,
the $x_i$s must all be different.
Thus, this gadget contains exactly $k-2$ colors,
and we can locate the positions of these $k-2$ colors, as desired.

\begin{figure}[h]
    \centering
    \includegraphics[scale=0.4]{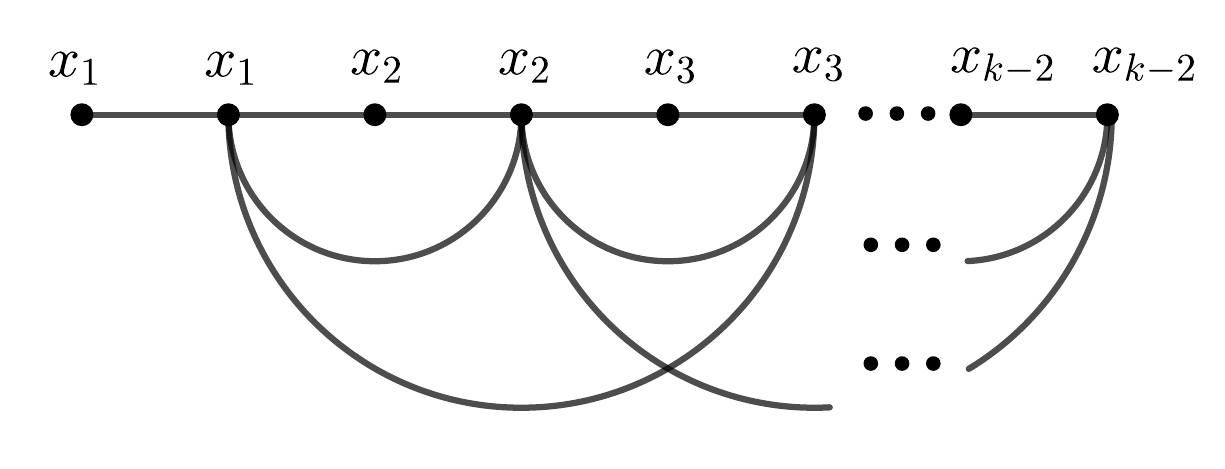}
    \caption{$(k-2)$-color gadget, consisting of exactly $k-2$ colors $x_1, x_2, \dots, x_{k-2}$ as shown, where $x_i$ can be any of the $k$ colors. Other nodes in the graph can connect to the even vertices of this gadget without interfering with its function.}
    \label{figure: k-color gadget}
\end{figure}

\subsection{Reduction}

Now we prove that $k$-colorable perfect matching is NP-hard by reducing
from 2-colorable perfect matching
(as proved by Schaefer \cite{schaefer1978complexity}).
Given a graph $G$, we create our $(k-2)$-color gadget,
and connect the $k-2$ nodes with the $k-2$ colors (the even nodes of the gadget)
to every node in $G$, as shown in Figure~\ref{figure: reducing k colors}.
Because each node in the gadget is already matched,
the nodes in $G$ cannot have the same color as any node in the gadget.
Thus, each node in $G$ must be colored with one of the two remaining colors.
This instance of $k$-colorable perfect matching is then equivalent
to 2-colorable perfect matching in~$G$.

\begin{figure}[h]
    \centering
    \includegraphics[scale=0.4]{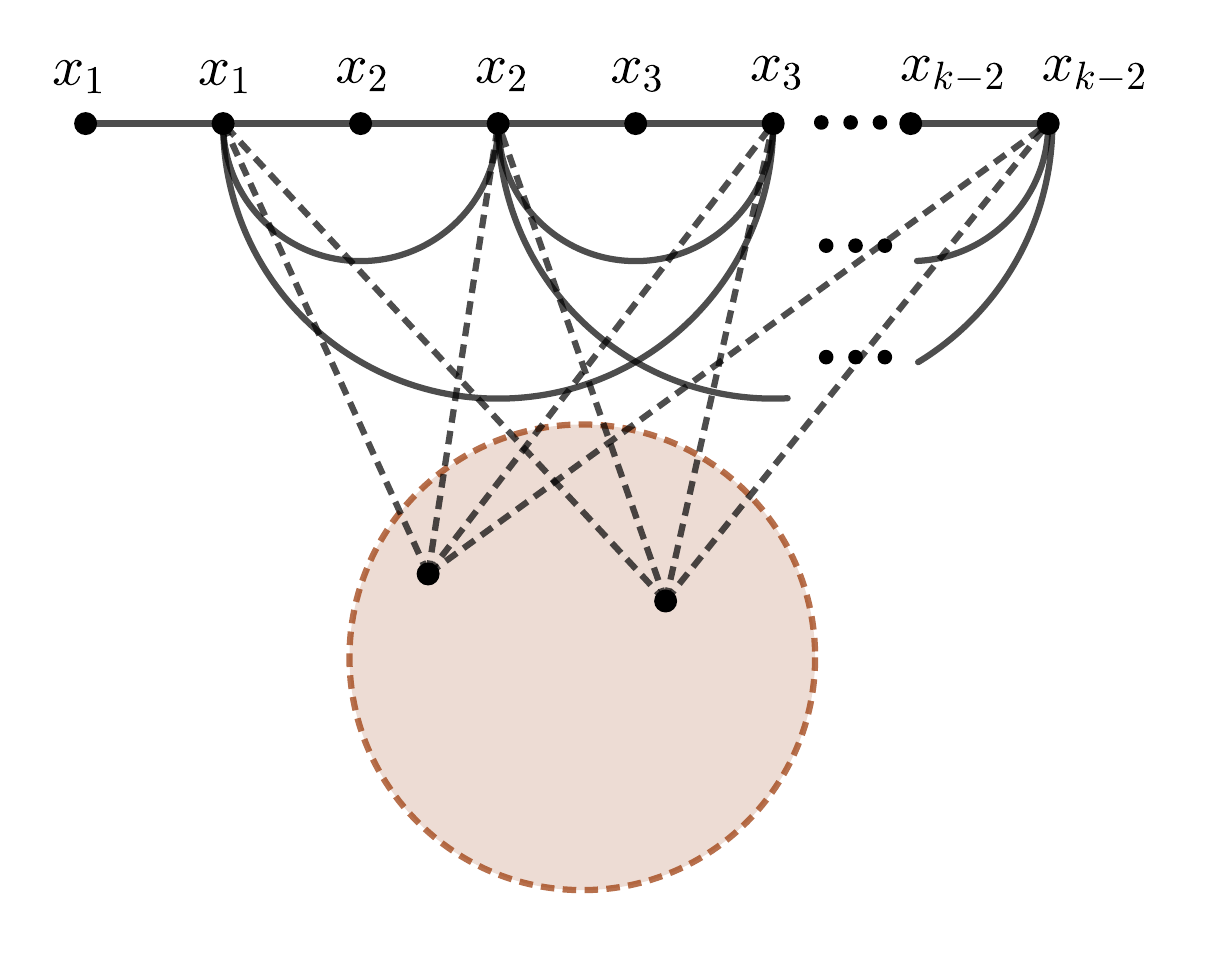}
    \caption{Connecting a $(k-2)$-color gadget (top path) with every node from the original graph $G$ (red circle), forcing each node to have only two possible color choices.}
    \label{figure: reducing k colors}
\end{figure}

\begin{theorem}
  $k$-colorable perfect matching in general graphs is NP-complete.
\end{theorem}

This reduction does not preserve 2-connectivity, 3-regularity, or planarity,
so it does not help to start from the stronger forms of 2-colorable perfect
matching from Section~\ref{section-planar}.
We leave as open problems whether $k$-colorable perfect matching remains
NP-complete when the graph is restricted to be 2-connected, 3-regular,
and/or planar.


\section*{Acknowledgments}

This work was initiated during open problem solving in the MIT class on
Algorithmic Lower Bounds: Fun with Hardness Proofs (6.892)
taught by Erik Demaine in Spring 2019.
We thank the other participants of that class
for related discussions and providing an inspiring atmosphere.

\bibliographystyle{alpha}
\bibliography{ref.bib}

\end{document}